%% file: main.tex
\documentclass[11pt]{article}
\usepackage{amstext}
\usepackage{amsmath}
\usepackage{amsthm}
\usepackage{amssymb}
\usepackage{bm}
\usepackage{wrapfig}
\usepackage{amsfonts, amsmath}
\usepackage[utf8]{inputenc}
\usepackage{graphicx}
\usepackage{bbm, comment}
\usepackage{subfigure}
\usepackage{color}
\usepackage{float}
\usepackage{wrapfig}
\usepackage{lipsum}
\usepackage{enumitem}
\usepackage{url}

\usepackage{thmtools,thm-restate}

\usepackage{hyperref}

\usepackage{tabulary}
\usepackage{booktabs}

\usepackage[top=1.in, bottom=1.in, left=1.in, right=1.in]{geometry}

\usepackage[numbers]{natbib}

\newcommand{\E}{\mathbb{E}}

\newtheorem*{theorem*}{Theorem}
\newtheorem{theorem}{Theorem}[section]
\newtheorem{lemma}[theorem]{Lemma}
\newtheorem{fact}[theorem]{Fact}

\newtheorem{example}[theorem]{Example}
\newtheorem{proposition}[theorem]{Proposition}
\newtheorem{definition}[theorem]{Definition}
\newtheorem{observation}[theorem]{Observation}

\newcommand\yk[1]{}
\newcommand\gs[1]{}

\title{False Consensus, Information Theory, and Prediction Markets\footnote{This is a write up of the theorem and proof which appeared in Kong and Schoenebeck’s EC 2017 tutorial \url{https://drive.google.com/open?id=18QifPXMezN42FgYnsYEgen7tjbxeuoEI}} }
\author{Yuqing Kong\thanks{This work is supported by National Natural Science Foundation of China award number~62002001}\\ Peking University \and Grant Schoenebeck\thanks{This work is supported by NSF CAREER 1452915 and NSF CCF 2007256.}\\ University of Michigan}

\date{}

\begin{document}

\maketitle

\begin{abstract} 
We study a setting where Bayesian agents with a common prior have private information related to an event's outcome and sequentially make public announcements relating to their information. Our main result shows that when agents' private information is independent conditioning on the event's outcome whenever agents have similar beliefs about the outcome, their information is aggregated.  That is, there is no false consensus. 

Our main result has a short proof based on a natural information theoretic framework. A key ingredient of the framework is the equivalence between the sign of the ``interaction information'' and a super/sub-additive property of the value of people's information. This provides an intuitive interpretation and an interesting application of the interaction information, which measures the amount of information shared by three random variables. 

We illustrate the power of this information theoretic framework by reproving two additional results within it: 1) that agents quickly agree when announcing (summaries of) beliefs in round robin fashion [Aaronson 2005]; and 2) results from [Chen et al 2010] on when prediction market agents should  release information to maximize their payment. We also interpret the information theoretic framework and the above results in prediction markets by proving that the expected reward of revealing information is the conditional mutual information of the information revealed.

\end{abstract}

\section{Introduction}
\input{intro}

\section{Preliminaries}

\input{prelimin}

\section{Information-theoretic Consensus}

\input{when_aggregation}

\section{Illuminating Prediction Markets with Information Theory}

\input{frame-work}
\subsection{Strategic Revelation}
\input{when_speak}

\section{Conclusion and Discussion}

\input{conclusion}

\newpage
\bibliographystyle{plainnat}
\bibliography{refs}

\newpage
\appendix
\section{Additional Proofs}\label{sec:add}
\input{appendix}

\section{Complete Agreement}\label{sec:completeagree}
\input{completeagreement}

\end{document}

%% file: intro.tex
Initially Alice thinks Democrats will win the next presidential election with probability 90\% and Bob thinks Democrats will win with 10\%. The players then alternate announcing their beliefs of the probability the Democrats will win the next presidential election. Alice goes first and declares, ``90\%''. Bob, then updates his belief rationally based on some commonly held information, some private information, and what he can infer from Alice's declaration (e.g. to 30\%) and announces that, ``30\%''.  Alice then updates her belief and announces it, and so forth. 

Formally, we have the following definition. 

\begin{definition} [Agreement Protocol \cite{aaronson2005complexity}] \label{protocol:agreement}
Alice and Bob share a common prior over the three random variables $W$, $X_A$, and $X_B$.  Here $W$ denotes the event to be predicted. Variables $X_A=x_A$ and $X_B=x_B$ are Alice's and Bob's  private information respectively, which can be from a large set and intricately correlated with each other and $W$.   

The agents alternate announcing their beliefs of $W$'s realization.

In round 1, Alice declares her rational belief $\mathbf{p}_A^1 = \Pr[W|X_A=x_A]$.  \footnote{Here and elsewhere we use the notation $\Pr[W]$ to denote a vector whose $w \in W$th coordinate indicates $\Pr[W=w]$. } 
Then Bob updates and declares his belief $\mathbf{p}_B^1 = \Pr[W|X_B=x_B,\mathbf{p}_A^1]$ rationally conditioning on his private information and what he can infer from Alice's declaration. 

Similarly, at round i, Alice announces her updated belief  $$\mathbf{p}_A^i=\Pr[W|X_A=x_A,\mathbf{p}_A^1,\mathbf{p}_B^1, \ldots, \mathbf{p}_A^{i-1},  \mathbf{p}_B^{i-1}];$$ and subsequently, Bob updates and announces his belief $$\mathbf{p}_B^i=\Pr[W|X_B=x_B,\mathbf{p}_A^1,\mathbf{p}_B^1, \ldots, \mathbf{p}_A^{i-1},  \mathbf{p}_B^{i-1}, \mathbf{p}_A^i].$$  

This continues indefinitely.
\end{definition}

% \begin{definition} [Agreement protocol \cite{aaronson2005complexity}]
%  Let $W$, $X_A$ and $X_B$ be three random variables. Alice and Bob share the same prior over them. Alice receives the realization $X_A=x_A$ privately and Bob receives the realization $X_B=x_B$ privately. The players then alternate announcing their beliefs of $W$'s realization. At round 1, Alice declares her rational belief $p_A^1$ where $\forall w, p_A^1(w)=\Pr[W=w|X_A=x_A] $. Then Bob updates and declares his belief $p_B^1$ where $\forall w, p_B^1(w)=\Pr[W=w|X_B=x_B,p_A^1]$ rationally conditioning on his private information and what he can infer from Alice's declaration. At round 2, Alice updates her belief to $p_A^2$ where $\forall w, p_A^2(w)=\Pr[W=w|X_A=x_A,p_A^1,p_B^1] $ rationally and announces it. Bob updates his belief to $p_B^2$ rationally and announces it, and so forth. 
% \end{definition}

Two fundamental questions arise from this scenario:
\begin{enumerate}
    \item  Will Alice and Bob ever agree or at least approximately agree, and if so will they (approximately) agree in a reasonable amount of time?
    \item If they (approximately) agree, will their agreement  (approximately) aggregate their information?  That is, will they (approximately) agree on the posterior belief conditioning on Alice and Bob's private information.  
\end{enumerate}

Aumann \cite{aumann1976agreeing} famously showed that rational Alice and Bob will have the same posterior belief given that they share the same prior and their posteriors are a common knowledge. In particular, if the agents in the agreement protocol ever stop updating their beliefs, they must agree.   While this may seem counter-intuitive, a quick explanation is that it is not rational for Alice and Bob to both persistently believe they know more than the the other person. This result does not fully answer the first question because the common knowledge requires a certain amount of time to be achieved.  

Both Aaronson \cite{aaronson2005complexity} and Geanakoplos and Polemarchakis \cite{geanakoplos1982we} answer the first question in the affirmative. Geanakoplos and Polemarchakis \cite{geanakoplos1982we} show that the agreement protocol will terminate after a finite number of messages.
Aaronson \cite{aaronson2005complexity} shows that without unbounded precision requirement, even when Alice and Bob only exchange a summary of their beliefs, rational Alice and Bob will take at most $O(\frac{1}{\delta \epsilon^2})$ rounds to have $(\epsilon,\delta)$-close beliefs,\footnote{$\Pr[|\text{Alice's expectation}-\text{Bob's expectation}|>
\epsilon]<\delta$ } regardless of how much they disagree with each other initially.

Alas, it is known the second question cannot always be answered in the affirmative, and thus agreement may not fully aggregate information.

\begin{example}[False consensus]
Say Alice and Bob each privately and independently flip a fair coin, and the outcome is the XOR of their results. Alice and Bob immediately agree, both initially proclaiming the probability  $0.5$.  However, this agreement does not aggregate their information; pooling their information, they could determine the outcome. 
\end{example}

Nonetheless, we answer the second question affirmatively for a large class of structures in a generalized context with more than two agents that we call the Round Robin Protocol. Notice that in the above false consensus example, Alice's and Bob's private information are independent but once we condition on the outcome they are dependent. Chen et al. \cite{chen2010gaming}  call this independent structure ``complements'' for reasons that will become clear. They also propose another independent structure, ``substitutes'' where both Alice and Bob's private information are independent conditioning on the outcome. Chen and Waggoner \cite{chen2016informational} further develop these concepts.

We will show that in the ``substitutes'' setting, i.e., when Alice and Bob's information are conditionally independent, (approximate) agreement implies (approximate) aggregation. We prove the results in the $n$ agents setting which is a natural extension of the Alice and Bob case. Our proof is direct and short based on information-theoretic tools. 

\paragraph{High Level Proof}  First, we denote the value of an agent's information as the mutual information between their private information and the outcome conditioning on the public information.  

The main lemma shows that the ``substitute'' structure implies the \emph{sub-additivity} of the values of people's private information.

When people approximately agree with each other conditioning on the history, the remaining marginal value of each individual's private information is small $\leq \epsilon$.  Under the ``substitutes'' structure, the sub-additive property of the main lemma implies the total value of information that has not been aggregated is most $n\epsilon$ where $n$ is the number of agents ($n=2$ in the Alice and Bob case). Therefore, (approximate) agreement implies (approximate) aggregation.

To show the main lemma, a key ingredient is the equivalence between the sign of the interaction information and the super/sub-additive property of the value of people's information. 
Interaction information is a generalized mutual information which measures the amount of information shared by three random variables. Unlike mutual information, interaction information can be positive or negative and thus is difficult to interpret and does not yet have a broad applications. Our framework provides an intuitive interpretation and an interesting application of the interaction information.

We additionally illustrate the power of this information theoretic framework by reproving two additional results within it: 1) that agents quickly agree when announcing beliefs in a round robin fashion \cite{aaronson2005complexity}; and 2) results from  Chen et al. \cite{chen2010gaming} that to maximize their payment in a prediction market, when signals are substitutes, agents should reveal them as soon as possible, and when signals are complements, agents should reveal them as late as possible.  We also interpret our information theoretic framework, our main result, and our quick convergence reproof in the context of prediction markets by proving that the expected reward of revealing information is the conditional mutual information of the information revealed.

The reproof that agents quickly agree uses the aggregated information as a potential function and observes that each round in which their is $\epsilon$ disagreement, the aggregated information must increase by $\epsilon$ (or a function of $\epsilon$ in the setting when agents only announce a summary of their beliefs). The result of when agents should reveal information in a prediction market follows from the sub/super-additivity of mutual information in each of these cases, which can be established using the sign of the interaction information.

\subsection{Related Work}

Protocols for consensus are well-studied in many different contexts with both Bayesian and non-Bayesian agents.  Many of these are in a growing field of social learning~\cite{golub2017learning}.

In some sense, the Bayesian update rule, also studied in this paper, is the most canonical and natural update rule.  The social learning literature concerning Bayesian agents typically asks questions about how Bayesian agents, each endowed with some private information, can aggregate their information using by public declarations.  When agents can only take binary (or a small number of) actions---often conceptualized as which of two products the agents is adopting---which depend on their beliefs, it is often discovered that agents can herd whereupon agents collectively have enough information to take more beneficial actions,  but fail to do so \cite{BHW1992,Banerjee1992-wr,Smith2000-fi}.  Our setting is different, because agents can act more then once.  However, herding is essentially a false consensus concerning a beneficial action.  In our setting, we ask a similar question about when the protocol can get stuck before aggregating the information of the agents. Other models in this literature look at agents embedded on networks~\cite{acemoglu2011opinion,lobel2016preferences,frongillo2011social} , that only generally announce binary information.  We do not consider a network structure as is often done in these works.

As already mentioned, prediction markets have also been analyzed in the context of Bayesian agents \cite{chen2010gaming,KongS2018AliceBobAlic,AnunrojwantCWX2019AliceBobAlice}.  We reprove one of the results of Chen et al. \cite{chen2010gaming}.  Like the result we reprove, these works study the optimal strategies for agents.  Kong and Schoenebeck \cite{KongS2018AliceBobAlic} show that sometimes even one bit of information can behave both like complements and substitutes: the agent would like to release part of it immediately, but part of it last.  

There is also a plethora of work studying non-Bayesian update rules and when consensus occurs or fails to occur---especially in the context of networks.  In this context it is sometimes not clear, whether agents are ``learning" or just trying to arrive at a consensus (e.g.~through imitation).  Typically, we think of agents as learning when there is a ground truth to which they are attempting to converge.  However, because the agents are non-Bayesian, the dynamics typically do not depend on the existence of a ground truth, and this part of the model is often not explicitly specified.

In these models, often the agents have a discrete state \cite{schoenebeck2018consensus,MosselS10,gao2019volatility}.  In such a case, to interpret the state as a belief, one has to rule out the granular beliefs of Bayesian reasoning.  When these models have continuous states, especially when the states are the $[0, 1]$ interval, it is easy to interpret the state as a belief.  However, the updates are based on some heuristic instead of being Bayesian.  For example, in the popular Degroot model~\cite{Degroot74}, agents update their state as a weighted average of their and their neighbors' signals in the previous round, and thus imitate their neighbors as opposed to the more delicate Bayesian reasoning.  In particular,  this update rule implies correlation neglect~\cite{enke2019correlation}---agents do not reason about how the information of their neighbors is linked.  Other models introduce edge weights that update\cite{durrett2012graph,gao2019volatility}, stubborn agents~\cite{yildiz2013binary,gao2017engineering}, and other modifications to more realistically model certain settings both intuitively and empirically~\cite{jackson2010social}.   As models become more attuned to predicting real agents, they also become more ad hoc and less canonical as the proliferation of models illustrates.  Instead, our paper operates in the most canonical model, understanding that this an imperfect model of human behavior, but nonetheless, can shed light on what does happen by aiding our understanding in this idealized model.

\paragraph{Independent Work}
Frongillo et al. \cite{frongillo2021agreement} independently prove that agreement implies aggregation of agents' information  under similar special information structures. However, the analyses are very different. Frongillo et al. \cite{frongillo2021agreement} employ a very delicate analysis that allows the results to be extended to general divergence measures. Our information-theoretic framework's analysis  provides a direct, short, and intuitive proof.

%% file: prelimin.tex
\subsection{Complements and Substitutes}
Following~\cite{chen2010gaming} we will be interested in two main types of signals.  

\begin{definition}(Substitutes and Complements~\cite{chen2010gaming})
$W$ denotes the event to be predicted.
\begin{description}
\item [Substitutes] Agents' private information $X_1,X_2,\cdots,X_n$ are independent conditioning on $W$. 
\item [Complements] Agents' private information $X_1,X_2,\cdots,X_n$ are independent.  
\end{description}
\end{definition}

\subsection{Information Theory Background}
This section introduces multiple concepts in information theory that we will use to analyze the consensus protocol and, later, prediction markets. 

\begin{definition}[Entropy \cite{shannon1948mathematical}]
We define the entropy of a random variable $X$ as
$$H(X):=-\sum_x \Pr[X=x]\log(\Pr[X=x]).$$

Moreover, we define the conditional entropy of $X$ conditioning on an additional random variable $Z=z$ as 
$$ H (X|Z=z):=-\sum_x\Pr[X=x|Z=z]\log(\Pr[X=x|Z=z]) $$
We also define the conditional entropy of $X$ conditioning on $Z$ as 
$$ H (X|Z):=\E_{Z}[H (X|Z=z)]. $$
\end{definition}

The entropy measures the amount of uncertainty in a random variable.  A useful fact is that when we condition on an additional random variable, the entropy can only decrease.  This follows immediately from the concavity of log.  

\begin{definition}[Mutual information \cite{shannon1948mathematical}]
We define the mutual information between two random variables $X$ and $Y$ as
$$I (X;Y):=\sum_{x,y}\Pr[X=x,Y=y]\log\left(\frac{\Pr[X=x,Y=y]}{\Pr[X=x]\Pr[Y=y]}\right)$$
Moreover, we define the conditional mutual information between random variables $X$ and $Y$ conditioning on an additional random variable $Z=z$ as
$$I (X;Y|Z=z):=\sum_{x,y}\Pr[X=x,Y=y|Z=z]\log\left(\frac{\Pr[X=x,Y=y|Z=z]}{\Pr[X=x|Z=z]\Pr[Y=y|Z=z]}\right)$$
We also define the conditional mutual information between $X$ and $Y$ conditioning on $Z$ as
$$I (X;Y|Z):=\E_{Z} I (X;Y|Z=z).$$
\end{definition}

\begin{fact}[Facts about mutual information]~\cite{cover2006elements}\label{cor:MI-prop}
\begin{description}
\item[Symmetry:] $I(X;Y) = I(Y;X)$
\item[Relation to entropy:] $H(X)-H(X|Y)=I(X;Y)$, $H(X|Z)-H(X|Y,Z)=I(X;Y|Z)$
\item[Non-negativity:] $I(X;Y)\geq 0$
\item[Chain rule:] $I(X,Y;Z)=I(X;Z)+I(Y;Z|X)$
\item[Monotonicity:] when $X$ and $Z$ are independent conditioning on $Y$, $I(X,Z)\leq I(X;Y)$.
\end{description}
\end{fact}

The first two facts follows immediately from the formula.  The third follows from the second, and the fact that conditioning can only decreases entropy.  The first three allow one to understand mutual information as the amount of uncertainly of one random variable that is eliminated once knowing the other (or vice versa).   The chain rule follows from the second because $I(X,Y;Z)= H(Z) - H(Z|X, Y) = H(Z) - H(Z | Y) + H(Z |Y)  - H(Z|X, Y) = I(X;Z)+I(Y;Z|X)$. The last property follows from the fact that $I(X; Y,Z)=I(X;Y)$ which can be proved by algebraic calculations and the chain rule which says $I(X; Y,Z) = I(X; Z) + I(X;Y|Z)$. 

\medskip

The interaction information, sometimes called co-information, is a generalization of the mutual information for three random variables. 

\begin{definition}[Interaction information \cite{ting1962amount}]
For three random variables $X$, $Y$ and $Z$, we define the interaction information among them as 
$$ I(X;Y;Z):=I(X;Y)-I(X;Y|Z) $$
\end{definition}

\begin{fact}[Symmetry~\cite{ting1962amount}]\label{fact:II-properties} 
 $I(X; Y; Z) = I(Y; X; Z)  =  I(X; Z; Y)$
\end{fact}

This can be verified through algebraic manipulations.

\paragraph{Venn Diagram}
As shown in figure~\ref{fig:diagram}, random variables $X,Y,Z$ can be visualized as sets $H(X), H(Y), H(Z)$ where the set's area represents the uncertainly of the random variable, and:  
\begin{description}
\item [Mutual Information:]  operation ``;" corresponds to intersection `` $\cap$" and is symmetric;
\item [Joint Distribution:]  operation ``," corresponds to union `` $\cup$" and is symmetric;
\item [Conditioning:]  operation `` $|$" corresponds to difference `` $|$";
\item [Disjoint Union:]  operation `` $+$" corresponds to the disjoint untion  `` $\sqcup$";
\end{description}
For example, $H(X, Y) = H(X)  + H(Y | X)$  because the LHS is  $H(X) \cup H(Y)$. Note that the interaction information corresponds to the center of the Venn diagram in figure~\ref{fig:diagram}.  However, despite the intuition that area is positive, the interaction information is not always positive.

\begin{figure}
\centering
\includegraphics[width=0.45\linewidth]{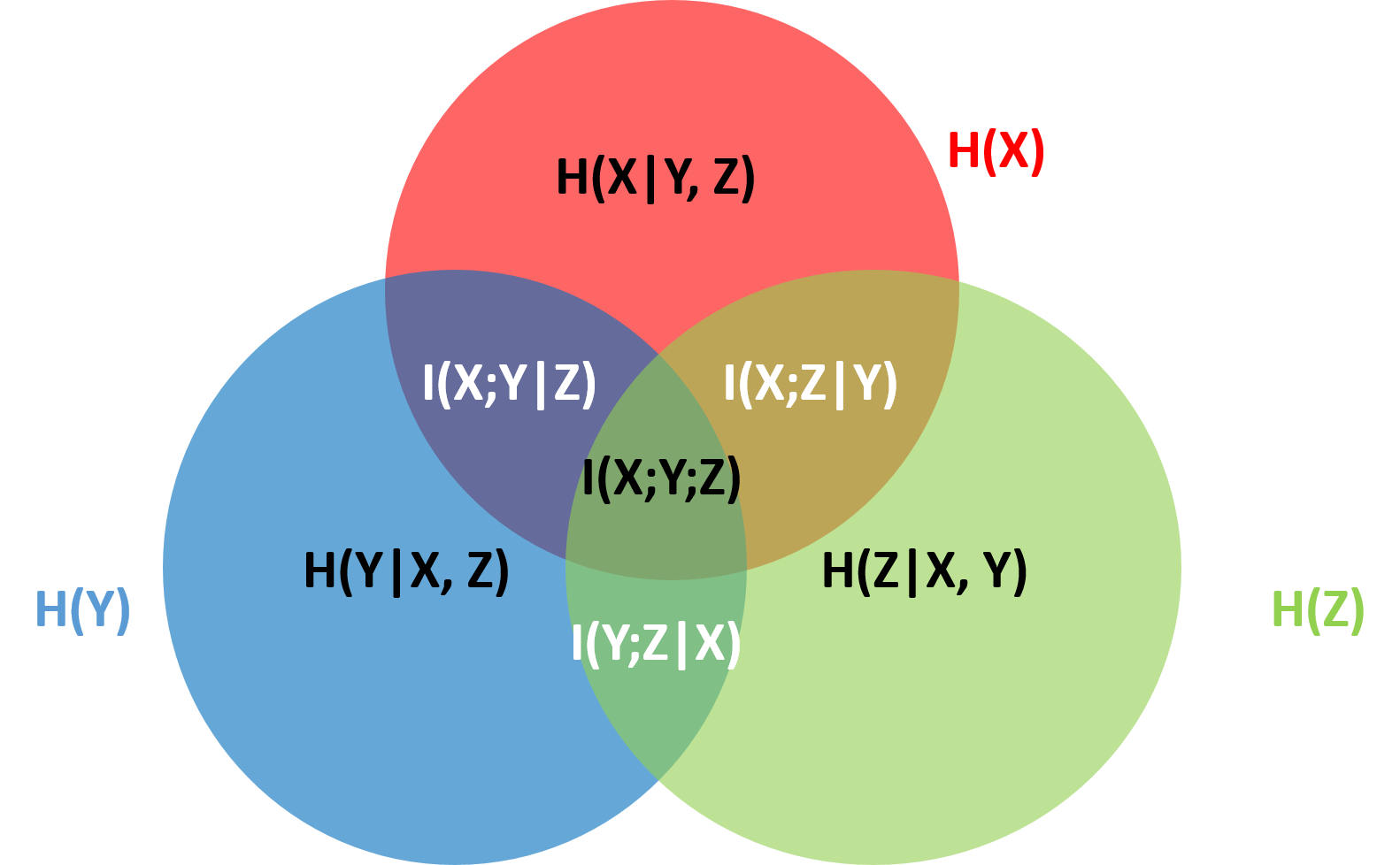}
\caption{Venn diagram}
\label{fig:diagram}
\end{figure}

%% file: when_aggregation.tex
We will analyze the Round Robin Protocol,  a multi-agent generalization of the Agreement Protocol, and use the information-theoretic framework to show that 1) consensus is quickly achieved; and 2)  the sub-additivity of ``substitutes'' guarantees that approximate consensus implies approximate aggregation. We first introduce the general communication protocol where agents make declarations in a round robin fashion.

\begin{definition} [Round Robin Protocol]  \label{protocol:consensus}
Agents  $A = \{1 \ldots, n\}$ share a common prior over the $n + 1$ random variables $W$, $X_1$, $X_2$, ..., $X_n$. Let $W$ denote the event to be predicted, and for $i \in A$, let $X_i=x_i$ denote agent $i$'s private message and its realization. 

In each round, the agents take turns sequentially announcing additional information.  In round $t$, agent 1 announces $h_1^t$ and then in sequence, for $i \in \{2, \ldots, n\}$,  agent $i$ announces $h_i^t$.  The first round is round 1 and then rounds proceed indefinitely. 

Let $H_i^t$ denote the history of declarations before agent $i$ announces $h_i^t$ in round $t$.  Thus $H_1^t = h_1^1,h_2^1, \ldots, h_{n}^{t-1}$ and for $i = 2$ to $n$, $H_i^t = h_1^1,h_2^1, \ldots, h_{i-1}^{t}$.  Also, it will be convenient to let $H^t = h_1^1,h_2^1, \ldots, h_n^1, \ldots, h_1^t,h_2^t, \ldots, h_n^t$ denote the history of the first $t$ rounds, and to interpret $H_{n + 1}^t$  as $H_{1}^{t+1}$. 
Note that $H^t = H_1^{t+1} = H_{n+1}^{t}$. 

Let $\mathbf{p}_i^t = \Pr[W|X_i=x_i, H_i^t]$ denote the belief of agent $i$ as she announces $h_i^t$ in round $t$.  Let $\mathbf{q}_i^t = \Pr[W|H_i^t]$ denote the belief of an outside observer (who knows the common prior but does not have any private message)  as agent $i$ announces $h_i^t$ in round $t$.

It is required that $h_i^t$ be well defined based on $X_i$ and $H_i^t$, that is, it should be defined on the filtration of information released up to that time.

\end{definition}

We state an information-theoretic definition for approximate consensus. With this definition the analysis of consensus time and false consensus becomes intuitive.

\begin{definition}[$\epsilon$-MI consensus] Round $t$ achieves $\epsilon$-MI consensus if for all $i$, $$I(X_i;W|H^t)\leq \epsilon.$$\end{definition}

We define the amount of information aggregated regarding $W$ at any time as the mutual information between $W$ and the historical declarations, $I(H;W)$. The following lemma shows two intuitive properties: 1) the amount of information aggregated is non-decreasing and 2) the growth rate depends on the marginal value of the agent's declaration. 

\begin{restatable}{lemma}{lemincreasing}\label{lem:increasing} \emph{(Information-theoretic Properties of the Protocol).} In the Round Robin Protocol, 
\begin{description}
\item [Non-decreasing Historical Information] any agent's declaration does not decrease the amount of information so $$I(H_{i+1}^t;W)\geq I(H_i^t;W),~\text{for all}~  i \in A, t \in \mathbb{N}.$$% \in \{1, \ldots, n-1\}, t$$ and $I(H_{1}^{t+1};W)\geq I(H_n^t;W),\forall t$.  
Therefore, the information does not decrease after each round: $$I(H^{t+1};W)\geq I(H^t;W),~\text{for all}~ t \in \mathbb{N}.$$
\item [Growth Rate = Marginal Value] the change in the historical information is the conditional mutual information between the acting agent's declaration and the predicted event conditioning on the history, i.e., $$I(H_{i+1}^t;W)- I(H_i^t;W)=I(h_i^t;W|H_i^t),~ \text{for all}~ i \in A, t \in \mathbb{N}.$$ 
\end{description}
\end{restatable}

These two properties directly follow from the properties of mutual information.  We defer the detailed proof to Appendix~\ref{sec:add}.

\subsection{Two Consensus Protocols}

We first introduce the natural Standard Consensus Protocol, which is the Agreement Protocol when there are two agents. We then introduce a discretized version of the protocol. 

\begin{definition} [Standard Consensus Protocol]  \label{protocol:standard}
The Standard Consensus Protocol is just the Round Robin Protocol where $h_i^t = \mathbf{p}_i^t$.  That is, each agent announces her belief.  
\end{definition}

\begin{figure}
\centering
\includegraphics[width=0.9\linewidth]{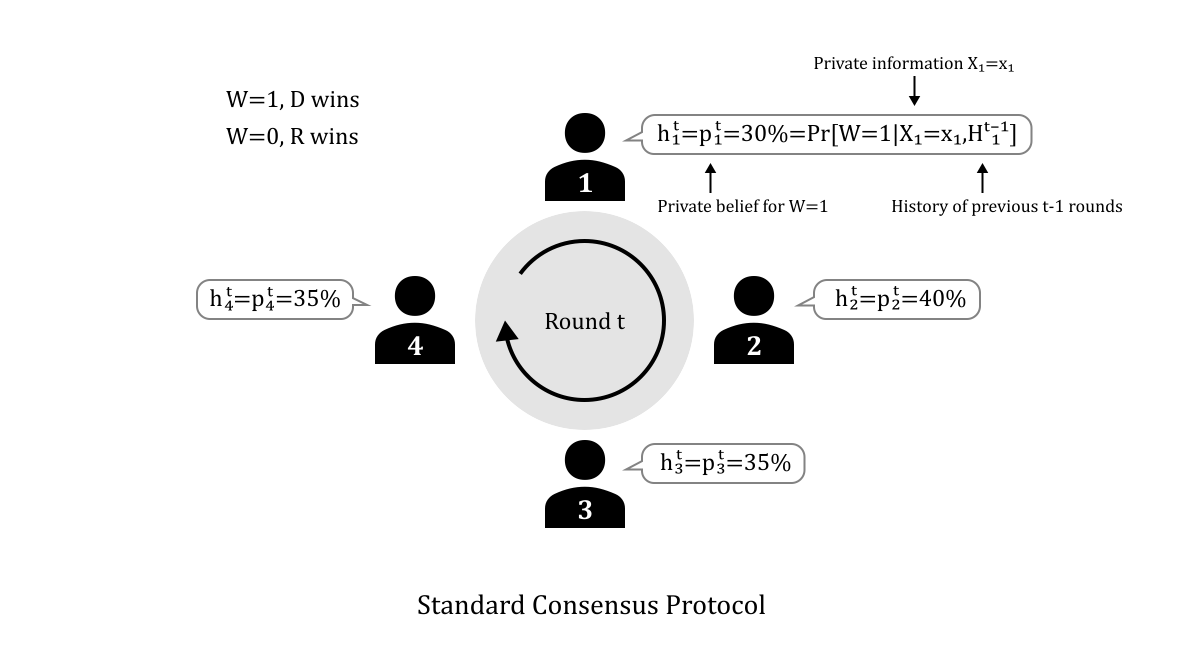}
\caption{Standard Consensus Protocol}
\label{fig:sp}
\end{figure}

In the above protocol, agents announce their beliefs. However, the number of bits required for belief announcement may be unbounded. Here we follow Aaronson \cite{aaronson2005complexity} and both focus on predicting binary events (occurs: 1; does no occur: 0) and discretize the protocol as follows: each agent  announces a summary of her current belief by comparing it to the belief of a hypothetical outsider who shares the same prior as the agents, sees all announcements, but does not possess any private information.  The agent announces  ``high'', if her belief that the probability the event will happen is far above the outsider's current belief; ``low'', if her belief is far below the outsider's current belief; and ``medium'' otherwise.   Note that all the agents have enough information to compute the outsider belief.

\begin{definition}[ $\epsilon$-Discretized Consensus Protocol]
Fix $\epsilon>0$.  Let $D_{KL}(p,q)= p\log\frac{p}{q}+(1-p)\log\frac{1-p}{1-q}$ be the KL divergence between a Bernoulli distribution $(1-p,p)$ and a  Bernoulli distribution $(1-q,q)$.  We use $p_i^t$ to denote agent i's belief for $W=1$ as she announces at round $t$ and $q_i^t$ to denote the hypothetical outsider's belief for $W=1$ at that time. We define $\overline{q_i^{t}}>q_i^t$ so that $D_{KL}(\overline{q_i^{t}},q_i^t)=\frac{\epsilon}{4}$, and $\underline{q_i^{t}}<q_i^t$ so that $D_{KL}(\underline{q_i^{t}},q_i^t)=\frac{\epsilon}{4}$.  Then agent i will announce a summary of her belief $$h_i^t=\begin{cases}\text{high}& p_i^t> \overline{q_i^{t}}\\\text{low}&p_i^t< \underline{q_i^{t}}\\\text{medium}&\text{otherwise.}\end{cases}$$ Note that each agent's output can be mapped to $\mathbf{R}$ by mapping high, medium, and low to 1, 0, and -1 respectively.  
\end{definition}

\begin{figure}
\centering
\includegraphics[width=0.9\linewidth]{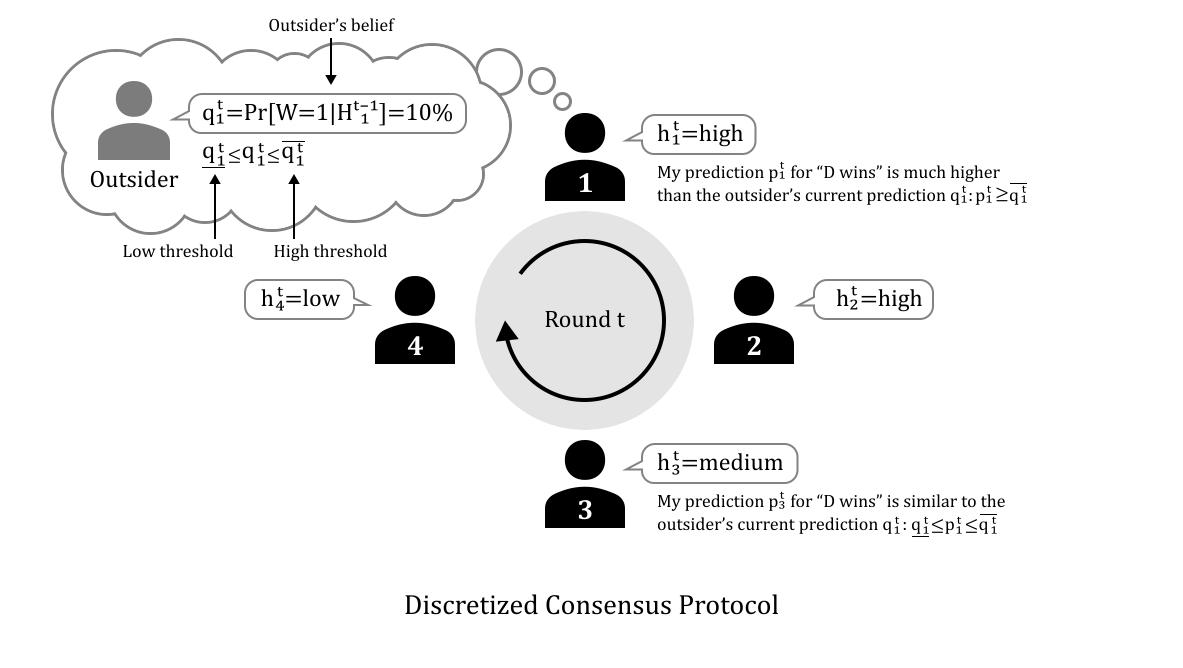}
\caption{Discretized Consensus Protocol}
%\caption{}
\label{fig:dp}
\end{figure}

This definition has the same conceptual idea as Aaronson \cite{aaronson2005complexity}, while the $\underline{q_i^{t}}$ and $\overline{q_i^{t}}$ are carefully designed so that informative private messages lead to informative declarations. We will state this property formally in the next section.

\subsection{Quick Consensus}\label{sec:quick-consensus}

In this section we show that both the Standard Consensus Protocol and the Discretized Consensus Protocol quickly reach $\epsilon$-MI consensus.  This result will not hold for general Round-Robin protocols.  For example, agents might not announce any useful information.\footnote{The reader has likely experienced a meeting like this.}

Lemma~\ref{lem:increasing} shows that the amount of aggregated information increases and the growth rate is the marginal value of the agent's declaration. Disagreement implies a $\geq \epsilon$ marginal value of the agent's private information.  We will state Lemma~\ref{lem:substantiallyincreasing} which shows that, in a case like this, where some agent has valuable private information, in the two defined protocols, this agent makes an informative declaration that will substantially increase the amount of aggregated information. This almost immediately leads to the quick consensus result because the total amount of aggregated information is bounded.

%Formally, Lemma~\ref{lem:substantiallyincreasing} shows that in both the Standard Consensus Protocol and the Discretized Consensus Protocol informative private messages lead to informative declarations.

%\begin{lemma}[Informative Declaration]\label{lem:substantiallyincreasing}
\begin{restatable}{lemma}{lemsubstantiallyincreasing}\label{lem:substantiallyincreasing} \emph{(Informative Declaration).} For all $i$ and $t$ and all possible histories $H_i^t$, when agent $i$'s private information is $X_i$ and $I(X_i;W|H_i^t)\geq \epsilon$, in the
 
 \begin{description}
 \item [Standard Consensus Protocol] we have $$I(h^t_i;W|H_i^t)=I(X_i;W|H_i^t)\geq \epsilon;$$
 \item [Discretized Consensus Protocol] we have $$I(h^t_i;W|H_i^t)\geq \frac{1}{64}\epsilon^3\frac{1}{\log\frac{1}{E^{-1}(\frac{\epsilon}{2})}} $$
where $E^{-1}(\epsilon)\leq 0.5$ and is the solution of $x\log x+(1-x)\log (1-x)=- \epsilon$.
 \end{description}
 
\end{restatable}

The above result in the Discretized Consensus Protocol requires delicate analysis (found in Appendix~\ref{sec:add}) and is mainly based on the fact that the mutual information is the expected KL divergence. 

\begin{theorem}[Convergence rate]\label{thm:quick}  The Standard Consensus Protocol achieves $\epsilon$-MI consensus in at most $\frac{2}{\epsilon}$ rounds. The Discretized Consensus Protocol achieves $\epsilon$-MI consensus in at most 
$\frac{512}{\epsilon^3}\log\frac{1}{E^{-1}(\frac14\epsilon)}$ %$512\left(\frac{\log\frac{1}{E^{-1}(\frac14\epsilon)}}{\epsilon^3}\right)$
rounds. 
\end{theorem}

\begin{proof}[Proof of Theorem~\ref{thm:quick}]
%To show the expected time when every agent $\epsilon$-agrees with the outsider is bounded, we use a natural potential function which is the amount of information aggregated. 

Because mutual information is monotone, the amount of information aggregated $I(H^t;W)$ is a non-decreasing function with respect to $t$. Moreover, we will show that when a round does not achieve $\epsilon$-MI consensus, $I(H^t;W)$ will have a non-trivial growth rate.  

Formally, when a round $t$ does not achieve $\epsilon$-MI consensus, there exists an agent $i$ such that $I(X_i;W|H^t)> \epsilon$. At round $t+1$, the expected amount of aggregated information increases at least the marginal value of the historical declarations $H^{t+1}_{i}$ before agent $i$ announces at round $t+1$, plus the marginal value of agent $i$'s declarations $h^{t+1}_{i}$.  This intuitively follows from monotonicity and the chain rule, and can be derived formally as follows:  \begin{align*}
I(H^{t+1};W)-I(H^t;W)\geq & I(H^{t+1}_{i+1};W)-I(H^t;W) \tag{$H^{t+1}$ contains $H^{t+1}_{i+1}$}\\ 
= & I(h^{t+1}_{i},H^{t+1}_{i};W)-I(H^t;W) \tag{$H^{t+1}_{i+1}=(H^{t+1}_{i},h^{t+1}_{i})$}\\
= & I(h^{t+1}_{i},H^{t+1}_{i};W|H^t)\tag{Based on Chain Rule and the fact that $H^{t+1}_{i}$ contains $H^t$}\\
= & I(H^{t+1}_{i};W|H^t)+I(h^{t+1}_{i};W|H^{t+1}_{i}) \tag{Again, Based on Chain Rule and the fact that $H^{t+1}_{i}$ contains $H^t$}
\end{align*} 
Thus, to analyze the growth rate, we need to show that when agent $i$'s private information is informative, $I(X_i;W|H^t)> \epsilon$, either agent $i$'s declaration $h^{t+1}_{i}$ is informative, or the historical declarations $H^{t+1}_{i}$ before agent $i$ announces at round $t+1$ is informative, conditioning the historical declarations $H^t$ before round $t+1$. 

When $I(X_i;W|H^t)> \epsilon$, we have $I(X_i,H^{t+1}_{i};W|H^t)> \epsilon$ as well. Moreover, because $$I(X_i,H^{t+1}_{i};W|H^t)=I(H^{t+1}_{i};W|H^t)+I(X_i;W|H^{t+1}_{i})$$ (Chain Rule), either 1) $I(H^{t+1}_{i};W|H^t)>\frac{\epsilon}{2}$ or 2) $I(X_i;W|H^{t+1}_{i})>\frac{\epsilon}{2}$. 

In case 1), the historical declarations $H^{t+1}_{i}$ is informative conditioning on $H^t$ and we have the growth rate $I(H^{t+1};W)-I(H^t;W)\geq I(H^t_{i};W|H^t)>\frac{\epsilon}{2}$. 

In case 2), for the Standard Consensus Protocol Lemma~\ref{lem:substantiallyincreasing} shows that  $I(h^{t+1}_{i};W|H^{t+1}_{i})=I(X_i;W|H^{t+1}_{i})>\frac{\epsilon}{2}$ because agent $i$ declares her Bayesian posterior as $h^{t+1}_{i}$. This again guarantees a $>\frac{\epsilon}{2}$ growth rate.  For the Discretized Consensus Protocol, Lemma~\ref{lem:substantiallyincreasing} shows that   $I(\psi(X_i);W|H^{t+1}_{i})\geq \frac{1}{512}\epsilon^3\frac{1}{\log\frac{1}{E^{-1}(\frac{\epsilon}{4})}}$.

Finally, for all $t$, $I(H^t;W)\leq I(X_1,X_2,...,X_n;W)\leq H(W)\leq \log_2 2 = 1$. Therefore, with at most $\frac{1}{\epsilon}$ rounds, the Standard Consensus Protocol achieves $\epsilon$-MI consensus, and with at most $\frac{512}{\epsilon^3}\log\frac{1}{E^{-1}(\frac14\epsilon)}$ %$512\left(\frac{\log\frac{1}{E^{-1}(\frac14\epsilon)}}{\epsilon^3}\right)$ 
rounds, the Discretized Consensus Protocol achieves $\epsilon$-MI consensus.  %\gs{can we remove $O$ here, seems unnecessary.}
\end{proof}

\subsection{No False Consensus with Substitutes} \label{sec:false-consensus}

We present our main result in this section: when agents' information is substitutes, there is no false consensus. A key ingredient is a subadditivity property for substitutes.  This result applies to general Round Robin protocols, and, in particular, is not restricted to the two specified consensus protocols.

In the ideal case, if we are given all agents' private information explicitly, the amount of information we obtain regarding $W$ is $I(X_1,X_2,\cdots,X_n;W)$. When agents follow the protocol for $t$ rounds, the amount of information we obtain regarding $W$ is $I(H^t;W)$. We care about the ``loss'', $I(X_1,X_2,\cdots,X_n;W)-I(H^t;W)$.

In the Standard Consensus Protocol under substitutes structure complete agreement is obtained after one round. The reason is that in such a situation, the information of each agent can be fully integrated after previous agents all precisely report their beliefs (see Appendix~\ref{sec:completeagree}). The following theorem shows a much more general results: under the substitutes structure \emph{every} protocol has the no false consensus property.

\begin{theorem}[Convergence $\Rightarrow$ Aggregation]\label{thm:agree}
    For all priors where agents' private information is substitutes, if agents achieve $\epsilon$-MI consensus after $t$ rounds, then the amount of information that has not been aggregated $I(X_1,X_2,\cdots,X_n;W)-I(H^t;W)$ is bounded by $n\epsilon$. 
\end{theorem}

We will use the following Lemma in the proof:

\begin{lemma}[Subadditivity for substitutes]\label{lemma:substitutes}
    When $X$ and $Y$ are independent conditioning on $Z$: 
\begin{description}
    \item[Nonnegative Interaction Information] $I(X;Y;Z)\geq 0$;
    \item[Conditioning Reduces Mutual Information] $I(Y;Z|X)\leq I(Y;Z)$;
    \item[Subadditivity of Mutual Information] $I(X,Y;Z)\leq I(X;Z)+I(Y;Z)$.
    \end{description}
    Moreover, when $X_1, \ldots, X_n$  are independent conditioning on $W$, $$ I(X_1, \ldots, X_n; W)\leq \sum_{i=1}^n I(X_i;W).$$
\end{lemma} 

\begin{proof}[Proof of Lemma~\ref{lemma:substitutes}]
Note that $I(X;Y|Z) = 0$ because $X$ and $Y$ are independent after conditioning on $Z$.  

Nonnegative Interaction Information follows because $I(X;Y;Z)=I(X;Y)-I(X;Y|Z) = I(X;Y) \geq 0$.  

Next $I(Y; Z|X) \geq I(Y; Z)$  follows after adding $I(Y; Z|X)$ to each item in the following derivation:  $0 \leq I(X;Y;Z) =  I(Y;Z;X) = I(Y; Z) - I(Y; Z|X)$  where the inequality is by  nonnegative interaction information, the first equality is from the symmetry of interaction information, and the second equality is from the definition of interactive information.

Next, subadditivity immediately follows because: $I(X, Y; Z) = I(X; Z) + I(Y; Z|X) \geq I(X;Z) + I(Y;Z)$ where the equality is from the chain rule and the inequality is because conditioning reduces mutual information. 

The moreover follows by using induction and subadditivity.  

\end{proof}

We require an additional observation that no history will disrupt the special information structure.

\begin{restatable}{observation}{obs}\label{obs}
%\begin{observation}\label{obs}
For any fixed history in the Round Robin Protocol, when $X_1,X_2,\cdots,X_n$ are substitutes (complements), they are still substitutes (complements) conditioning on any history. 
\end{restatable}

We defer the proof to Appendix~\ref{sec:add}. Aided by the above information-theoretic properties, we are ready to show our direct and short proof.

\begin{proof}[Proof for Theorem~\ref{thm:agree}]
If after $t$ rounds the Round Robin Protocol achieves $\epsilon$-MI consensus, then, by definition, for all $i$, $$I(X_i;W|H^{t})\leq \epsilon.$$  After $t$ rounds, the amount of information we have aggregated is $I(H^t;W)$. We can aggregate at most $I(X_1,X_2,...,X_n;W)$ information when all agents' private information is revealed explicitly. Thus, the amount of information that has not been aggregated is 
	
	\begin{align*}
	& I(X_1,X_2,...,X_n;W)-I(H^t;W)\\ \tag{The history $H^t$ is a function of $X_1, \ldots, X_n$}
	=&I(X_1,X_2,...,X_n,H^t;W)-I(H^t;W)\\ \tag{Chain rule}
	=& I(X_1,X_2,...,X_n;W|H^t)\\ \tag{Sub-additivity}
	 \leq & \sum_i I(X_i;W|H^t)\leq n\epsilon
	\end{align*} 
	
	The inequality is from subadditivity: note that by Observation~\ref{obs}  because $X_1, \ldots, X_n$ are independent after conditioning on $W$,  $X_1, \ldots, X_n$ are still independent after conditioning on $H$.     
\end{proof}

%% file: frame-work.tex
\subsection{Prediction Market Overview}

A prediction market is a place to trade information. For example, a market maker can open a prediction market for the next presidential election with two kinds of shares, $D$-shares and $R$-shares. If Democrats wins, each $D$-share pays out one dollar but $R$-shares are worth nothing. If Republicans win, each $R$-share is worth one dollar and $D$-shares are worth nothing. People reveal their information by trading those shares. For example, if an agent believes Democrats will win with probability 0.7, he should buy D-shares as long as its price is strictly lower than \$ 0.7. If people on balance believe that the current price is less than the probability with which the event will occur, the demand for shares (at that price) will outstrip the supply, driving up the price.  Similarly if the price is too high they will buy the opposite share as the prices should sum to 1 since exactly one of the two will payout 1.  

 Hanson \cite{hanson2003combinatorial,hanson2012logarithmic} proposes a model of prediction markets that is theoretically equivalent to the above which we describe below.  Instead of buying/selling shares to change the price, the agents simply change the market price directly.  We define this formally below.

%\gs{Should we first define as a strategy is ideal if the price after each round is equal the probability?    Then have a characterization in terms of partial revealing?}

\subsection{Preliminaries for Prediction Markets}

We introduce prediction markets formally and relate them to the Standard Consensus Protocol. We focus on prediction markets which measure the accuracy of a prediction using the logarithmic scoring rule.  

\begin{definition} (Logarithmic scoring rule~\cite{https://doi.org/10.1111/j.2517-6161.1952.tb00104.x})
Fix an outcome space $\Sigma$ for a signal $\sigma$.  Let $\mathbf{q} \in \Delta_{\Sigma}$ be a reported distribution. 

The Logarithmic Scoring Rule maps a signal and reported distribution to a payoff as follows:
$$L(\sigma,\mathbf{q})=\log (\mathbf{q}(\sigma)).$$
\end{definition}

\begin{definition}[Market scoring rule ~\cite{hanson2003combinatorial,hanson2012logarithmic,chen2012utility}]
    We build a market for random variable  $W$ as follows:  the market sets an initial belief $\mathbf{p}_0$ for  $W$.  A sequence of agents is fixed.  Activity precedes in rounds.  In the $i$th round, the corresponding agent can change the market price from $\mathbf{p}_i$ to $\mathbf{p}_{i+1}$, and will be compensated   $L(W,\mathbf{p}_{i+1})-L(W,\mathbf{p}_{i})$ after $W$ is revealed. 
\end{definition}

Let the signal $\sigma$ be drawn from some random process with distribution $\mathbf{p} \in \Delta_\Sigma$.

Then the expected payoff of the Logarithmic Scoring Rule is:

\begin{align} 
\E_{\sigma \leftarrow \mathbf{p}}[L(\sigma,\mathbf{q})]=\sum_{\sigma}\mathbf{p}(\sigma)\log \mathbf{q}(\sigma)=L(\mathbf{p},\mathbf{q})
\end{align}

It is well known (and easily verified) that this value will be uniquely maximized if and only if $\mathbf{q}=\mathbf{p}$.  Because of this, the logarithmic scoring rule is called a \emph{strictly proper scoring rule}.

\paragraph{Agent Strategies}  We would like to require agents to commit to rational strategies.  In such a case, while agents may hide information, they will not try to maliciously trick other agents by misreporting their beliefs.  Formally, we assume that each agent $i$ plays a strategy of the following form: they change the market price to $\Pr[W=w|S_h(X_i),H = h]$ where $H$ is the history and $S_h$ is a possibly random function of $X_i$ that is allowed to depend on the history. That is, she declares the rational belief conditioning on the realization of $S_h(X_i)$ and what she can infer from the history. 

A natural question is whether we should expect fully strategic agents to behave maliciously rather than just hiding information.  Previous work~\cite{AnunrojwantCWX2019AliceBobAlice} studies a restricted setting and shows that misreporting or tricking other agents does not happen in equilibrium even when agents are allowed to. 

Two additional observations: 1) the agent's strategy might not actually reveal $S_h(X_i)$ because it could be that $\Pr[W=w|S_h(x_i),H = h] = \Pr[W=w|S_h(x_i'),H = h]$ for $x_i \neq x_i'$.  2) If the agents' always use $S_h(X_i) = X_i$, and update the market price based on their entire message, then this essentially reduces to the Standard Consensus Protocol.  In this case, the agents are playing myopically and optimizing their payment at each step.

\subsection{Expected Payment = Conditional Mutual Information}

Amazingly, we show that the expected payment to an agent, is just a function of the conditional mutual information of the random variables that he reveals. This connects the expected payment in prediction markets and amount of information.

\begin{lemma}[Expected payment = conditional mutual information]  \label{lemma:payment_is_conditionMI}
In a prediction market with the log scoring rule, the agent who changes the belief from $\Pr[W|H]$ to $\Pr[W|X=x,H]$ for all $x$ will be paid $I(X;W|H)$ in expectation. 
\end{lemma}

%Aidded by chain rule and the relation to entropy, we obtain the following corollary. 

%\begin{corollary}In a prediction market with log scoring rule, if one agent changes the belief from $\Pr[W|\Pi]$ to $\Pr[W|X=x,\Pi]$ and the other agent changes the belief from $\Pr[W|X=x,Y=y,\Pi]$, in total they will be paid $I(X,Y;W|\Pi)$ in expectation. Moreover, the total amount of expected payment paid by the market is bounded by $H(W)$. \end{corollary}

\begin{proof} [Proof of Lemma~\ref{lemma:payment_is_conditionMI}] Fix any history $H$:
\begin{align*}
    &\E_{X, W| H} L(W,\Pr[W|H,X])-L(W,\Pr[W|H])\\
    &=\sum_{W,X} \Pr[W,X|H] \log\left(\frac{\Pr[W|H,X]}{\Pr[W|H]}\right)\\
    &= I(X;W|H)
\end{align*}
\end{proof}

\paragraph{Interpretation of Previous Results} When agents participate in the market one by one and update the market price to $\Pr[W=w|S_h(X_i),H = h]$ , the list of market prices can be interpreted as running a Standard Consensus Protocol.  The connection between the expected payment and conditional mutual information leads to the following interpretations. 

\begin{description}
\item [$\epsilon$-MI Consensus] If round $t$ achieves $\epsilon$-MI consensus, the expected payment of any particular agent obtained by changing the market price to her Bayesian posterior is at most  $\epsilon$. \item [Quick Consensus] For any round $t$ that does not achieve $\epsilon$-MI consensus, in round $t+1$ at least one agent will obtain at least an $\epsilon$ payment in expectation. However, Lemma~\ref{lemma:payment_is_conditionMI} also implies that the total expected payment in the prediction market is bounded by the entropy of $W$. Thus, at most $\frac{H(W)}{\epsilon}$ rounds are needed for $\epsilon$-MI consensus in the prediction market. 
\item [No False Consensus with Substitutes] Even when the market price has noise and agents only reveal a summary of their beliefs by participating in the markets, if round $t$ achieves $\epsilon$-MI consensus, no agent has information with value greater than $\epsilon$.  Then subadditivity implies that currently, the expected payment of all agents' private information together is bounded by $n\epsilon$, thus is not valuable. 
% When $\epsilon$-MI consensus achieves, for all agents, her expected payment via changing the market belief to her best belief is bounded by $\epsilon$. The subadditivity implies that currently, the expected payment of all agents' private information together is bounded by  $n\epsilon$, thus is not valuable. 
\end{description}

%% file: when_speak.tex
This section will use the above information-theoretic properties to provide an alternative proof for the results proved in Chen et al. \cite{chen2010gaming}.

\begin{definition}[Alice-Bob-Alice (ABA)]
	Alice and Bob hold private information $X_A,X_B$ related to event $W$ respectively. There are three stages. Alice can change the market belief at stage 1 and stage 3. Bob can change the market belief at stage 2. 
\end{definition}

We assume that the strategic players can hide their information but not behave maliciously. A strategy profile is a profile of the players' strategies.  Moreover, like the original paper, we assume that an agent not only reports their belief $\Pr[W=w|S_h(X_i),H = h]$  but also reveals their information $S_h(X_i)$.  \footnote{This addresses knife-edge situations where $\Pr[W=w|S_h(x_i),H = h] = \Pr[W=w|S_h(x_i'),H = h]$ for $S_h(x_i) \neq S_h(x_i')$.  That is, in the XOR case, if Bob reveals his full signal, then he will not only report the belief as $\frac12$ but also announce his full signal. }

%In prediction market based on log scoring rule, if Alice reveals all her private information at stage 1 by changing the price vector from $p^0$ where $\forall w, p^0(w)=\Pr[W=w]$ to $p_A^1$ where $\forall w, p_A^1(w)=\Pr[W=w|X_A]$ and announcing $X_A$, she will be paid $I(X_A;W)$ in expectation and if she reveals all her private information at stage 3 by changing the price vector from $p_B^2$ where $\forall w, p_B^2(w)=\Pr[W=w|X_B]$ to $p_A^3$ where $\forall w, p_A^3(w)=\Pr[W=w|X_A,X_B]$ and announcing $X_A$, she will be paid $I(X_A;W|X_B)$ in expectation. If Alice reveals part of her information at stage 1, say $S_1(X_A)$, and part of her information at stage 3, say $S_3(X_A)$, she will be paid $I(S_1(X_A);W)+I(S_3(X_A);W|X_B,S_1(X_A))$. 

\begin{proposition} [ABA \& Information structures \cite{chen2010gaming}]\label{prop:aba}
In prediction market based on logarithmic scoring rule,
	\begin{description}
\item[Substitutes] when Alice and Bob's private information are substitutes, the strategy profile where Alice reveals $X_A$ at stage 1 and Bob reveals $X_B$ at stage 2 is an equilibrium;
\item[Complements] when Alice and Bob's private information are complements, the strategy profile where Bob reveals $X_B$ at stage 2 and Alice reveals $X_A$ at stage 3 is an equilibrium. 
\end{description}
\end{proposition}

%\begin{lemma}[Superadditivity for complements]\label{lemma:complements}

We will use the following lemma, which is a direct analogue of Lemma~\ref{lemma:substitutes}

\begin{restatable}{lemma}{lemmacomplements}\label{lemma:complements} \emph{(Superadditivity for complements).} When $X$ and $Y$ are independent: 
\begin{description}
    \item[Nonpositive Interaction Information] $I(X;Y;Z)\leq 0$;
    \item[Conditioning Increases Mutual Information] $I(Y;Z|X)\geq I(Y;Z)$;
    \item[Superadditivity of Mutual Information] $I(X,Y;Z)\geq I(X;Z)+I(Y;Z)$.
    \end{description}
    Moreover, when $X_1, \ldots, X_n$  are independent conditioning on $W$, $$ I(X_1, \ldots, X_n; W)\geq \sum_{i=1}^n I(X_i;W).$$
\end{restatable} 

The proof is directly analogous to that of Lemma~\ref{lemma:substitutes} and we defer it to Section~\ref{sec:add}.

\begin{proof}[Proof of Proposition~\ref{prop:aba}]
	First, for Bob, to maximize his expected payment, it's always optimal to reveal $X_B$ in stage 2 because Bob is paid for his conditional mutual information which is maximized by full revelation. It's left to analyze Alice's optimal strategy given that Bob reveals $X_B$ in stage 2. 
	
	If Alice reveals all her information either at the first stage or at the third stage, we only need to compare $I(X_A;W)$ and $I(X_A;W|X_B)$.  But $I(X_A;W)-I(X_A;W|X_B)=I(X_A;X_B;W)$. Thus, the results immediately follow from the fact that the sign of the interaction information is nonnegative/nonpositive when the information are substitutes/complements.
	
	It's left to consider the general strategy where Alice reveals part of her information at stage 1, say $S_1(X_A)$, and part of her information at stage 3, say $S_3(X_A)$. In this case, she will be paid $I(S_1(X_A);W)+I(S_3(X_A);W|X_B,S_1(X_A))$ according to Lemma~\ref{lemma:payment_is_conditionMI}. 
	
	First, it is optimal for Alice to reveal all her remaining information at the last stage since $I(S_3(X_A);W|X_B,S_1(X_A))\leq I(X_A;W|X_B,S_1(X_A))$ due to monotonicity of mutual information. Thus, Alice will reveal $S_3(X_A)=X_A$.

	Under the substitutes structure, Alice's expected payment in stage 3 is $$I(X_A;W|X_B,S_1(X_A))\leq I(X_A;W|S_1(X_A))$$ because 1) fixing any $S_1(X_A)$, $X_B$ and $X_A$ are still independent conditioning on $W$ thus are substitutes (Observation~\ref{obs}) and 2) conditioning decreases mutual information for substitutes. Therefore, 
	
\begin{align*}
	&I(S_1(X_A);W)+I(X_A;W|X_B,S_1(X_A))\\ \tag{conditioning decreases mutual information}
	\leq &I(S_1(X_A);W)+I(X_A;W|S_1(X_A))\\ \tag{chain rule}
	= &I(X_A,S_1(X_A);W)= I(X_A;W).
\end{align*}
	
Thus $I(X_A;W)$ upper bounds Alice's total payment, but this is what she receives if she reveals all her information in the first round.  
		
	Under complements structure, Alice's expected payment in stage 1 is $$I(S_1(X_A);W)\leq I(S_1(X_A);W|X_B)$$ because 1) $S_1(X_A)$ and $X_B$ are independent thus are complements and 2) conditioning decreases mutual information for complements. Therefore, 
	
	\begin{align*}
		&I(S_1(X_A);W)+I(X_A;W|X_B,S_1(X_A))\\ \tag{conditioning increases mutual information}
		\leq &I(S_1(X_A);W|X_B)+I(X_A;W|X_B,S_1(X_A))\\ \tag{chain rule}
		=&I(X_A,S_1(X_A);W|X_B)=I(X_A;W|X_B)
	\end{align*}
	
Thus $I(X_A;W|X_B)$ upper bounds Alice's total payment, but this is what she receives if she reveals all her information in the second round.  
	
\end{proof}

%% file: conclusion.tex
We have developed an information theoretic framework to analyze the aggregation of information both in the Round Robin protocol and prediction markets. We showed that when agents' private information about an event is independent conditioning on the event's outcome, then, when the agents are in near consensus in any Round Robin Protocol, their information is nearly aggregated.  We additionally reproved 1) the Standard/Discretized Consensus Protocol quickly converges [Aaronson 2005] in the information-theoretic framework; and 2) results from Chen et al. \cite{chen2010gaming}  on when prediction market agents should  release information to maximize their payment. 
%Some of these results can be slightly extended.  One can using any proper scoring rule, not just the logarithmic scoring rule to define a prediction market.  

Our analysis of the Alice Bob Alice prediction market straightforwardly extends to any sequence of an arbitrary number of agents.  By applying the same argument, it is easily shown that: 1) every agent revealing all their information immediately is an equilibrium in the substitutes case; and 2) every agent revealing all their information as late as possible is an equilibrium in the complements case.

One possible extension is to study similar protocols when the agents are on networks.  This was already pioneered by  Parikh and Krasucki \cite{parikh1990communication} and also examined by Aaronson~\cite{aaronson2005complexity} who used a spanning tree structure to show that a particular agreement protocol converges quickly.  Perhaps using the tools of this paper, one could analyze more general protocols.  

Another possible extension is to look at generalizations of Shannon mutual information.  For example, starting with any strictly proper scoring rule one can develop a Bregman mutual information~\cite{Kong:2019:ITF:3309879.3296670} and ask whether our proofs will go through using this new mutual information definition.  
When using the logarithmic scoring rule, one arrives at the standard Shannon mutual information used in this paper.  However, different scoring rules are possible as well.  All such mutual informations will still obey the chain rule and non-negativity. As such, generalized versions of Lemma~\ref{lem:increasing} will hold.  However they may not be symmetric (and similarly their interactive information may not be symmetric).  Thus, our techniques cannot be straightforwardly adjusted to reprove Theorem~\ref{thm:agree} or Proposition~\ref{prop:aba} in this  manner. So while this generalizes our framework, finding a good application remains future work.

We hope that our framework can analyze additional settings of agents aggregating information.  One example would be more inclusive classes of agent signals than in the settings of this paper. However, perhaps our framework could also be applied to analyze broader settings such as social learning~\cite{golub2017learning} or rewards for improvements in machine learning outcomes~\cite{abernethy2011collaborative}, where either individually developed machine learning predictors are eventually combined in ensembles or the training data is augmented by individually procured training data. 

%rewarding individual machine learning algorithms which are eventually combined in ensembles or data contributed toward learning to perform some task, or even social learning. 

%The proof that the consensus/agreement protocol converges quickly will generalize to any such scoring rule, though this is not too interesting as it just allows one to redefine what ``agreement'' and ``aggregation'' mean.     

%Our result for the aggregation of information relies on the symmetry of the interaction information as does our reproof of when agents should reveal information in a prediction market.  This is not known to hold for variants of mutual information that are not based on the logarithmic scoring rule.

%Finally, our reproof of when agents should reveal information in a prediction market relies on the chain rule and sub/superadditivity.  The sub/superadditivity properties really only relied when $MI(X; Y|Z) = 0$  or $MI(X; Y) = 0$.  And these same facts will be true for any Bregman mutual information in the complements/substitutes settings.  

%We do note however, that our analysis of the Alice Bob Alice prediction market trivially extends to any sequence of an arbitrary number of agents.  By repeatedly applying the same argument, it is easily shown that, assuming other agents are using a hiding strategy, an agent can improve his utility by revealing all her information at the beginning for substitution or in their final round of participation for complements.   

%Look at prediction market stuff that was removed.

%% file: appendix.tex
\lemincreasing*

\begin{proof}[Proof of Lemma~\ref{lem:increasing}]
Because the (conditional) mutual information is always non-negative, once we have established  the second property of the growth rate equaling the marginal value the first property of non-decreasing information follows immediately.  Thus, we start by proving the second property, the growth rate equals the marginal value.   

\begin{align}\label{mar}
I(H_{i+1}^t;W)-I(H_{i}^t;W)=I(H_{i+1}^t;W|H_{i}^t)= I(h_i^t;W|H_i^t).
\end{align}

The first equality holds due to chain rule and the fact that $H_{i+1}^t$ contains $H_{i}^t$. The second equality holds due to the fact that $H_{i+1}^t=(H_{i}^t,h_i^t)$. 

\end{proof}

\lemsubstantiallyincreasing*

\begin{proof}[Proof of Lemma~\ref{lem:substantiallyincreasing}]
In the Standard Consensus protocol, $I(h^t_{i};W|H^{t}_{i})=I(X_i;W|H^{t}_{i})\geq \epsilon$ because agent $i$ declares her Bayesian posterior as $h^{t}_{i}=\mathbf{p}_i^t=\Pr[W|X_i=x,H_i^t=h]$. 
%In the Discretized Consensus Protocol, we first illustrate the high level idea. 

\medskip

\emph{Discretized Consensus Protocol Proof Summary:} We will first obvserve that if the average over all possible histories $I(X_i;W|H_i^t)\geq \epsilon$, then there must be a set of histories with non-trivial weights such that $I(X_i;W|H_i^t=h)> \frac{\epsilon}{2}$. Fixing a history $h$, the summary function maps the agent's private information to three signals: high, low, and medium. This classifies the expectations conditioning on the private information $x$ into three categories: the high set, the low set, and the medium set. We will first show that when the private information is informative, because the medium set's expectations are close to the outsider's current expectations, the medium set's contribution to $I(X_i;W|H_i^t=h)$ will not be significant. Thus either the high set or the low set contributes a lot. We then prove that after compression,  the high set and the low set will still preserve a non-trivial amount of information. We will repeatedly use the fact that the mutual information is the expected KL divergence in the analysis.

\medskip

We first show that if $I(X_i;W|H_i^t)\geq \epsilon$ then, $\Pr_h[I(X_i;W|H_i^t=h) > \epsilon/2] \geq \epsilon/2$.  This follows from Markov's Inequality applied to $1 - I(X_i;W|H_i^t=h)$ which is always non-negative because $I(X_i;W|H_i^t=h)\leq H(W)\leq 1$. Formally, 

\begin{align*}
\Pr_h[I(X_i;W|H_i^t=h) \leq \epsilon/2] = & \Pr_h[1- I(X_i;W|H_i^t=h) \geq 1- \epsilon/2]\\ \tag{Markov's Inequality}
\leq & \frac{1-I(X_i;W|H_i^t)}{1- \epsilon/2}\\
\leq & \frac{1-\epsilon}{1- \epsilon/2}\leq 1- \epsilon/2
\end{align*}

\gs{removed the calculations?} \yk{I change the calulations}

% When $I(X_i;W|H_i^t)\geq \epsilon$, we have $\sum_h \Pr[H_i^t=h] I(X;W|H_i^t=h)\geq \epsilon$. 

% Moreover, 

% \begin{align*}
% &\sum_h \Pr[H_i^t=h] I(X;W|H_i^t=h)\\
% =&\sum_{I(X;W|H_i^t=h)< \frac12\epsilon} \Pr[H_i^t=h] I(X;W|H_i^t=h)+\sum_{I(X;W|H_i^t=h)\geq \frac12\epsilon} \Pr[H_i^t=h] I(X;W|H_i^t=h)\\
% < & \frac12\epsilon+ \sum_{I(X;W|H_i^t=h)\geq \frac12\epsilon} \Pr[H_i^t=h] I(X;W|H_i^t=h)
% \end{align*}

% Therefore, we have $\sum_{I(X;W|H_i^t=h)\geq \frac12\epsilon} \Pr[H_i^t=h] I(X;W|H_i^t=h)> \frac12\epsilon$. This implies that \[\sum_{I(X;W|H_i^t=h)\geq \frac12\epsilon} \Pr[H_i^t=h] > \frac12\epsilon\] because $I(X;W|H_i^t=h)\leq H(W)\leq \log_2  2=1$.

\medskip

We now fix any history $h$ where $I(X_i;W|H_i^t=h) \geq \epsilon/2$.  We would like to show that \begin{align}
    I(h_i^t;W|H_i^t=h) \geq \frac{\epsilon^2}{32}\frac{1}{-\log_2 E^{-1}(\frac 12\epsilon)} \label{eq:goal}
\end{align} because then

\begin{align*}
& I(h_i^t;W|H^{t}_{i})\\
\geq &\sum_{I(X_i;W|H_i^t=h)\geq \frac{\epsilon}{2}} I(h_i^t;W|H_i^t=h) \cdot  \Pr[H_i^t=h] \\
\geq &\frac{\epsilon^2}{32}\frac{1}{-\log_2 E^{-1}(\frac{\epsilon}{2})}  
 \cdot \Pr_h\left[I(X_i;W|H_i^t=h)\geq\frac{\epsilon}{2}\right] \\ % \sum_{I(X;W|H_i^t=h)\geq\frac{\epsilon}{2}} \Pr[H_i^t=h]\\
\geq &\frac{\epsilon^2}{32}\frac{1}{- \log_2 E^{-1}(\frac{\epsilon}{2})} \cdot  \frac{\epsilon}{2}\\
= &\frac{1}{64 }\epsilon^3\frac{1}{-\log_2 E^{-1}(\frac{\epsilon}{2})}
\end{align*}

and this proves the lemma.

Before showing Equation~\ref{eq:goal}, for notational clarity, we define $$ \psi(X_i) = h_i^t =\begin{cases}\text{high}& p_i^t> \overline{q_i^{t}}\\\text{low}&p_i^t< \underline{q_i^{t}}\\\text{medium}&\text{otherwise}.\end{cases}$$ Notice that once we fix a history, $\psi(X_i) = h_i^t$ is a function of $X_i$. Recall that we have defined $\overline{q_i^{t}}>q_i^t$ so that $D_{KL}(\overline{q_i^{t}},q_i^t)=\frac{\epsilon}{4}$, and $\underline{q_i^{t}}<q_i^t$ so that $D_{KL}(\underline{q_i^{t}},q_i^t)=\frac{\epsilon}{4}$.

Additionally, let $q$ be a shorthand for $q_i^t=\Pr[W=1|H_i^t=h]$. Let $p_x = \Pr[W=1|H_i^t=h,X_i=x]$ be the Bayesian posterior for $W=1$ conditioning on that agent $i$ receives $X_i=x$. Let $w_{x}=\Pr[X_i=x|H_i^t=h]$ be the prior probability that agent receives $X_i=x$. 
 Let $p_{hi}=\Pr[W=1|H_i^t=h,\psi(X_i)=\text{high}]$ be the Bayesian posterior for $W=1$ conditioning on that the agent $i$ announces `high'. Let $w_{hi}=\Pr[\psi(X_i)=\text{high}|H_i^t=h]$ be the prior probability that agent $i$ announces `high'. Analogously, we define $p_{lo}$, $w_{lo}$ (low), $p_{me}$, and $w_{me}$  (medium). 

Our goal in Equation~\ref{eq:goal} can then be restated as $$I(\psi(X_i);W|H_i^t=h) \geq  \frac{\epsilon^2}{32}\frac{1}{-\log_2 E^{-1}(\frac{\epsilon}{2})}.$$

Notice that: 
\begin{align}
& I(X_i;W|H_i^t=h) \\= &\nonumber  \sum_x \Pr[X_i=x|H_i^t=h] \sum_w \Pr[W=w|X_i=x,H_i^t=h]\log_2  \frac{\Pr[W=w|X_i=x,H_i^t=h]}{\Pr[W=w|H_i^t=h]}\nonumber \\
=& \sum_x \Pr[X_i=x|H_i^t=h] D_{KL}(p_x,q) \nonumber \\
= & \sum_x w_x D_{KL}(p_x,q) \label{eq:KL-sum}
\end{align} 

We can partition all $x$ into three categories, $\psi(x)=high,low, medium$.  
Because when $\psi(x)=medium$, $D_{KL}(p_x,q)\leq \frac{\epsilon}{4}$, when $I(X;W|H_i^t=h)\geq \frac{\epsilon}{2}$, we have \[\sum_{\psi(x)=high,low} w_x D_{KL}(p_x,q)\geq \frac{\epsilon}{4}.\] Thus either the low set or the high set contributes $\geq \frac{\epsilon}{8}$. Without loss of generality, we assume $\sum_{\psi(x)=high} w_x D_{KL}(p_x,q)\geq \frac{\epsilon}{8}$. 

Recall our goal is to show that given Equation~\ref{eq:KL-sum} was greater than $\epsilon/ 2$ then the following is large:

\begin{align}
I(\psi(X_i);W|H_i^t=h) =   \sum_{\psi(X_i) \in \{lo, me, hi\}}  w_{\psi(X_i)} \cdot D_{KL}(p_{\psi(X_i)}, q) 
                       \geq    w_{hi} \cdot D_{KL}(p_{hi}, q). \label{eq:high}
\end{align}

We can show this is large by lower-bounding both  $ w_{hi}$ and $ D_{KL}(p_{hi}, q)$.  

First, we will lower bound $ w_{hi}$ by upper bounding $D_{KL}(p_x,q)$.  Note that $$D_{KL}(p, q) \leq \max\{\log \frac{1}{q}, \log \frac{1}{1-q} \} = \max\{ -\log q, -\log 1-q \} $$ by recalling the formula $D_{KL}(p, q) = p \log \frac{p}{q} + (1-p) \frac{1-p}{1-q}$.  To upper bound this, we must show that $q$ is not too close to 0 or 1. Because $q\log_2 \frac{1}{q}+(1-q)\log_2  \frac{1}{1-q}=I(W;W|H_i^t=h)\geq I(X;W|H_i^t=h)\geq \frac 12\epsilon$, we have $E^{-1}(\frac {\epsilon}{2} )\leq q\leq 1- E^{-1}(\frac{\epsilon}{2} )$.  This gives us that  $D_{KL}(p_x,q) \leq  -\log_2 E^{-1}(\frac {\epsilon}{2})$.

We assumed that $\sum_{\psi(x)=high} w_x \cdot D_{KL}(p_x,q)\geq \frac{1}{8}\epsilon$, so we have:

$$  w_{hi} = \sum_{\psi(x)=high} w_x \geq \frac {\epsilon}{8} \frac{1}{- \log_2 E^{-1}(\frac{\epsilon}{2} )}.$$

Next, we lower bound $ D_{KL}(p_{hi}, q)$.  
$ D_{KL}(p_{hi}, q) \geq \frac{\epsilon}{4}$ because  for any $x$ where $\psi(x) = high$, we have that $p_x \geq  \overline{q}$ and $\overline{q}$ was defined so that  $ D_{KL}(\overline{q},q)  = \frac{\epsilon}{4}$ and for any $q' > \overline{q}$, $ D_{KL}(q',q) > \frac{\epsilon}{4}$.

Combining with Equation~\ref{eq:high} we are now done because

\begin{align*}
I(\psi(X_i);W|H_i^t=h) \geq  w_{hi} \cdot D_{KL}(p_{hi},q) \geq  \frac{\epsilon^2}{32}\frac{1}{-\log_2 E^{-1}(\frac{\epsilon}{2})}
\end{align*}

\end{proof}

\lemmacomplements*

\begin{proof}[Proof of Lemma~\ref{lemma:complements}]
The proof is directly analogous to that of Lemma~\ref{lemma:substitutes}

First, $I(X;Y) = 0$ because $X$ and $Y$ are independent.

Nonpositive Interaction Information follows because $I(X;Y;Z)=I(X;Y)-I(X;Y|Z) = - I(X;Y|Z) \leq 0$.  

Next $I(Y; Z|X) \geq I(Y; Z)$  because  $0 \geq I(X;Y;Z) =  I(Y;Z;X) = I(Y; Z) - I(Y; Z|X)$ where the inequality is by  nonpostive interaction information, the first equality is from the symmetry of interaction information, and the second equality is from the definition of interactive information.

Third, superadditivity immediately follows because: $I(X, Y; Z) = I(X; Z) + I(Y; Z|X) \geq I(X;Z) + I(Y;Z)$ where the equality is from the chain rule and the inequality is because conditioning increases mutual information. 

The moreover follows by using induction and superadditivity.  

\end{proof}

\obs*

\begin{proof}[Proof of Observation~\ref{obs}]
Let $D$ be some distribution over $\Sigma^n$ where $X_1,X_2,\cdots,X_n$ are all independent.
We first observe that after any ``independent'' restrictions on the realizations of $X_1,X_2,\cdots,X_n$, they are still independent.  That is, fix $\Sigma_1, \ldots, \Sigma_n$ where $\Sigma_i \subseteq \Sigma$ for all $i \in \{1, \ldots, n\}$ and let $\xi \subset \Sigma^n$ be the event where $x_i \in \Sigma_i$ for all $i \in \{1, \ldots, n\}$.  Then conditioning on $\xi$, $X_1,X_2,\cdots,X_n$ are independent as well.

For all $(x_1,x_2,\cdots,x_n)\in\Sigma_1\times\Sigma_2\times\cdots \Sigma_n$, 

\begin{align*}
\Pr_D[X_1=x_1,X_2=x_2,\cdots,X_n=x_n|\xi] &= \frac{\Pr_D[X_1=x_1,X_2=x_2,\cdots,X_n=x_n]}{\Pr_D[\xi]}\\
&=  \frac{\Pi_i \Pr_D[X_i=x_i]}{\Pi_i \Pr_D[X_i\in\Sigma_i]}\\
&=  \Pi_i \Pr_D[X_i=x_i|X_i\in\Sigma_i ]
\end{align*}

% \begin{align*}
% &\Pr[X_1=x_1,X_2=x_2,\cdots,X_n=x_n|\forall i, X_i, X_i\in\Sigma_i] \\=& \frac{\Pr[X_1=x_1,X_2=x_2,\cdots,X_n=x_n]}{\Pr[\forall i, X_i, X_i\in\Sigma_i]}\\
% =& \frac{\Pr[X_1=x_1,X_2=x_2,\cdots,X_n=x_n]}{\Pr[\forall i, X_i, X_i\in\Sigma_i]}\\
% = & \frac{\Pi_i \Pr[X_i=x_i]}{\Pi_i \Pr[X_i\in\Sigma_i]}\\
% = & \Pi_i \Pr[X_i=x_i|X_i\in\Sigma_i ]
% \end{align*}
Moreover, 
$\Pr_D[X_i=x_i| X_i\in \Sigma_i] =  \Pr_D[X_i=x_i|\xi]$ because 

\begin{align*}
&\Pr_D[X_i=x_i|\xi]\\
&= \frac{\Pr_D[X_i=x_i,\xi]}{\Pr_D[\xi]}\\
&=  \frac{\Pr_D[X_i=x_i,\forall j\neq i, X_j\in\Sigma_i]}{\Pr_D[\forall j,  X_j\in\Sigma_j]}\\
&=  \frac{\Pr_D[X_i=x_i]\Pi_{j \neq i}\Pr_D[X_j \in \Sigma_j]}{\Pr_D[X_i\in \Sigma_i]\Pi_{j \neq i}\Pr_D[X_j \in \Sigma_j]}\\
&=  \frac{\Pr_D[X_i=x_i]}{\Pr_D[X_i \in \Sigma_i]}\Pi_{j \neq i}\frac{\Pr_D[X_j \in \Sigma_j]}{\Pr_D[X_j\in \Sigma_j]}\\
&=  \frac{\Pr_D[X_i=x_i]}{\Pr_D[X_i\in \Sigma_i]}\\
&=\Pr_D[X_i=x_i| X_i\in \Sigma_i] 
\end{align*}

% \begin{align*}
% &\Pr[X_i=x_i|\forall j, X_j, X_j\in\Sigma_i]\\
% =& \frac{\Pr[X_i=x_i,\forall i, X_i, X_i\in\Sigma_i]}{\Pr[\forall j, X_j, X_j\in\Sigma_i]}\\
% = & \frac{\Pr[X_i=x_i,\forall j, X_j, X_j\in\Sigma_i]}{\Pr[\forall j, X_j, X_j\in\Sigma_i]}\\
% = & \frac{\Pr[X_i=x_i,\forall j\neq i, X_j, X_j\in\Sigma_i]}{\Pr[\forall j, X_j, X_j\in\Sigma_i]}\\
% = & \frac{\Pr[X_i=x_i]}{\Pr[X_i\in \Sigma_i]}
% \end{align*}

Therefore, we have $\Pr[X_1=x_1,X_2=x_2,\cdots,X_n=x_n|\xi]=\Pi_i \Pr[X_i=x_i|\xi ]$. For all $(x_1,x_2,\cdots,x_n)\notin\Sigma_1\times\Sigma_2\times\cdots \Sigma_n$, we have $\Pr[X_1=x_1,X_2=x_2,\cdots,X_n=x_n|\xi]=0=\Pi_i \Pr[X_i=x_i|\xi]$.  Thus conditioning on $\xi$, $X_1,X_2,\cdots,X_n$ are independent as well.

In the Round Robin Protocol, after the first agent makes her declaration $h_1^1$ which is a function of the prior and $X_1$, we have a restriction for $X_1$. The second agent's declaration $h_2^1$ is a function of the prior, $h_1^1$ and $X_2$. For fixed $h_1^1$, we have an independent restrictions for $X_2$, and so forth. Thus, for any fixed history, we will have independent restrictions for $X_1$, $X_2$, ..., and $X_n$. Therefore, when $X_1$, $X_2$, ..., and $X_n$ are independent, they are still independent given any fixed history as well.

The same analysis shows that when $X_1$, $X_2$, ..., and $X_n$ are independent conditioning any $W=w$, they are still independent conditioning on both $W=w$ and any fixed history. 

\end{proof}

%% file: completeagreement.tex
\begin{observation}
For all priors where agents' private information are substitutes, in the standard consensus protocol, after one round, every agent's belief becomes $\Pr[W=1|X_1=x_1,X_2=x_2,\cdots, X_n =x_n]$. 
\end{observation}

The above observation shows that when agents' private information are substitutes, complete agreement is obtained after one round in the standard consensus protocol. 

\begin{proof}
We use $p^0$ to denote the prior $\Pr[W=1]$.  We use $\ell^0$ to denote the prior likelihood $\ell^0 = \frac{\Pr[W=1]}{\Pr[W=0]}$

The key observation is that there is a bijection $\ell_0 = \frac{p_0}{1-p_0}$ between the probability and likelihood space.  To see it is a bijection, note that $p_0 = \frac{\ell}{1 + \ell}$.  

In the standard consensus protocol, at round 1, agent 1 reports $p_1^1=\Pr[W=1|X_1=x_1]$.   But, because there is a bijection between the probability and likelihood space, they might just report the likelihood:
$$\ell_1^1=\frac{\Pr[W=1|X_1=x_1]}{\Pr[W=0|X_1=x_1]} = l_0 \cdot \frac{\Pr[X_1=x_1|W=1]}{\Pr[X_1=x_1|W=0]}$$

The second equality follows from applying Bayes rule to the numerator and denominator.  

Similarly, because the signals are conditionally independent:

\begin{align*} \ell_1^2 = & \frac{\Pr[W=1|X_1=x_1,X_2=x_2]}{\Pr[W=0|X_1=x_1,X_2=x_2]}\\
=& \frac{\Pr[W=1,X_1=x_1,X_2=x_2]}{\Pr[W=0,X_1=x_1,X_2=x_2]}\\ \tag{Conditional independence}
=& \frac{\Pr[W=1]\Pr[X_1=x_1|W=1]\Pr[X_2=x_2|W=1]}{\Pr[W=0]\Pr[X_1=x_1|W=0]\Pr[X_2=x_2|W=0]}\\
=& \frac{\Pr[W=1|X_1=x_1]\Pr[X_2=x_2|W=1]}{\Pr[W=0|X_1=x_1]\Pr[X_2=x_2|W=0]}\\
= &\ell_1^1 \cdot \frac{\Pr[X_2=x_2|W=1]}{\Pr[X_2=x_2|W=0]}.
\end{align*}

\yk{expand this part: done}

%where again, the second inequality follows by applying Bayes rule to the numerator and denominator. 

Analogously, for $i\geq 3$, each agent $i$ simply updates the likelihood by multiplying by $\frac{\Pr[X_i=x_i|W=1]}{\Pr[X_i=x_i|W=0]}$.  While the are actually reporting a probability, because of the bijection, it is equilivlent that they report a likelihood. 

Notice, that the likelihood at the end of one round, captures all the agent's information.

\end{proof}

%% file: main.bbl
\begin{thebibliography}{34}
\providecommand{\natexlab}[1]{#1}
\providecommand{\url}[1]{\texttt{#1}}
\expandafter\ifx\csname urlstyle\endcsname\relax
  \providecommand{\doi}[1]{doi: #1}\else
  \providecommand{\doi}{doi: \begingroup \urlstyle{rm}\Url}\fi

\bibitem[Aaronson(2005)]{aaronson2005complexity}
Scott Aaronson.
\newblock The complexity of agreement.
\newblock In \emph{Proceedings of the Thirty-Seventh Annual ACM Symposium on
  Theory of Computing}, pages 634--643. ACM, 2005.

\bibitem[Abernethy and Frongillo(2011)]{abernethy2011collaborative}
Jacob~D Abernethy and Rafael Frongillo.
\newblock A collaborative mechanism for crowdsourcing prediction problems.
\newblock \emph{Advances in Neural Information Processing Systems}, 24, 2011.

\bibitem[Acemoglu and Ozdaglar(2011)]{acemoglu2011opinion}
Daron Acemoglu and Asuman Ozdaglar.
\newblock Opinion dynamics and learning in social networks.
\newblock \emph{Dynamic Games and Applications}, 1\penalty0 (1):\penalty0
  3--49, 2011.

\bibitem[Anunrojwong et~al.(2019)Anunrojwong, Chen, Waggoner, and
  Xu]{AnunrojwantCWX2019AliceBobAlice}
Jerry Anunrojwong, Yiling Chen, Bo~Waggoner, and Haifeng Xu.
\newblock Computing equilibria of prediction markets via persuasion.
\newblock In \emph{International Conference on Web and Internet Economics},
  pages 45--56. Springer, 2019.

\bibitem[Aumann(1976)]{aumann1976agreeing}
Robert~J Aumann.
\newblock Agreeing to disagree.
\newblock \emph{The annals of statistics}, pages 1236--1239, 1976.

\bibitem[Banerjee(1992)]{Banerjee1992-wr}
A~V Banerjee.
\newblock A simple model of herd behavior.
\newblock \emph{Q. J. Econ.}, 107\penalty0 (3):\penalty0 797--817, 1992.

\bibitem[Bikhchandani et~al.(1992)Bikhchandani, Hirshleifer, and
  Welch]{BHW1992}
Sushil Bikhchandani, David Hirshleifer, and Ivo Welch.
\newblock A theory of fads, fashion, custom, and cultural change as
  informational cascades.
\newblock \emph{Journal of political Economy}, 100\penalty0 (5):\penalty0
  992--1026, 1992.

\bibitem[Chen and Pennock(2007)]{chen2012utility}
Yiling Chen and David~M Pennock.
\newblock A utility framework for bounded-loss market makers.
\newblock In \emph{23rd Conference on Uncertainty in Artificial Intelligence
  (UAI)}, 2007.
\newblock Most recent version: arXiv preprint arXiv:1206.5252, 2012.

\bibitem[Chen and Waggoner(2016)]{chen2016informational}
Yiling Chen and Bo~Waggoner.
\newblock Informational substitutes.
\newblock In \emph{Foundations of Computer Science (FOCS), 2016 IEEE 57th
  Annual Symposium on}, pages 239--247. IEEE, 2016.

\bibitem[Chen et~al.(2010)Chen, Dimitrov, Sami, Reeves, Pennock, Hanson,
  Fortnow, and Gonen]{chen2010gaming}
Yiling Chen, Stanko Dimitrov, Rahul Sami, Daniel~M Reeves, David~M Pennock,
  Robin~D Hanson, Lance Fortnow, and Rica Gonen.
\newblock Gaming prediction markets: Equilibrium strategies with a market
  maker.
\newblock \emph{Algorithmica}, 58\penalty0 (4):\penalty0 930--969, 2010.

\bibitem[Cover and Thomas(2006)]{cover2006elements}
Thomas~M. Cover and Joy~A. Thomas.
\newblock \emph{Elements of Information Theory 2nd Edition (Wiley Series in
  Telecommunications and Signal Processing)}.
\newblock Wiley-Interscience, July 2006.
\newblock ISBN 0471241954.

\bibitem[DeGroot(1974)]{Degroot74}
M.H. DeGroot.
\newblock Reaching a consensus.
\newblock \emph{Journal of the American Statistical Association}, pages
  118--121, 1974.

\bibitem[Durrett et~al.(2012)Durrett, Gleeson, Lloyd, Mucha, Shi, Sivakoff,
  Socolar, and Varghese]{durrett2012graph}
Richard Durrett, James~P Gleeson, Alun~L Lloyd, Peter~J Mucha, Feng Shi, David
  Sivakoff, Joshua~ES Socolar, and Chris Varghese.
\newblock Graph fission in an evolving voter model.
\newblock \emph{Proceedings of the National Academy of Sciences}, 109\penalty0
  (10):\penalty0 3682--3687, 2012.

\bibitem[Enke and Zimmermann(2019)]{enke2019correlation}
Benjamin Enke and Florian Zimmermann.
\newblock Correlation neglect in belief formation.
\newblock \emph{The Review of Economic Studies}, 86\penalty0 (1):\penalty0
  313--332, 2019.

\bibitem[Frongillo et~al.(2021)Frongillo, Neyman, and
  Waggoner]{frongillo2021agreement}
Rafael Frongillo, Eric Neyman, and Bo~Waggoner.
\newblock Agreement implies accuracy for substitutable signals.
\newblock \emph{arXiv preprint arXiv:2111.03278}, 2021.

\bibitem[Frongillo et~al.(2011)Frongillo, Schoenebeck, and
  Tamuz]{frongillo2011social}
Rafael~M Frongillo, Grant Schoenebeck, and Omer Tamuz.
\newblock Social learning in a changing world.
\newblock In \emph{International Workshop on Internet and Network Economics},
  pages 146--157. Springer, 2011.

\bibitem[Gao et~al.(2017)Gao, Li, Schoenebeck, and Yu]{gao2017engineering}
Jie Gao, Bo~Li, Grant Schoenebeck, and Fang{-}Yi Yu.
\newblock Engineering agreement: The naming game with asymmetric and
  heterogeneous agents.
\newblock In Satinder Singh and Shaul Markovitch, editors, \emph{Proceedings of
  the Thirty-First {AAAI} Conference on Artificial Intelligence, February 4-9,
  2017, San Francisco, California, {USA}}, pages 537--543. {AAAI} Press, 2017.
\newblock URL \url{http://aaai.org/ocs/index.php/AAAI/AAAI17/paper/view/14986}.

\bibitem[Gao et~al.(2019)Gao, Schoenebeck, and Yu]{gao2019volatility}
Jie Gao, Grant Schoenebeck, and Fang-Yi Yu.
\newblock The volatility of weak ties: Co-evolution of selection and influence
  in social networks.
\newblock In \emph{Proceedings of the 18th International Conference on
  Autonomous Agents and MultiAgent Systems}, pages 619--627, 2019.

\bibitem[Geanakoplos and Polemarchakis(1982)]{geanakoplos1982we}
John~D Geanakoplos and Heraklis~M Polemarchakis.
\newblock We can't disagree forever.
\newblock \emph{Journal of Economic theory}, 28\penalty0 (1):\penalty0
  192--200, 1982.

\bibitem[Golub and Sadler(2016)]{golub2017learning}
Benjamin Golub and Evan Sadler.
\newblock Learning in social networks.
\newblock \emph{The Oxford Handbook of the Economics of Networks}, 2016.

\bibitem[Good(1952)]{https://doi.org/10.1111/j.2517-6161.1952.tb00104.x}
I.~J. Good.
\newblock Rational decisions.
\newblock \emph{Journal of the Royal Statistical Society: Series B
  (Methodological)}, 14\penalty0 (1):\penalty0 107--114, 1952.

\bibitem[Hanson(2003)]{hanson2003combinatorial}
Robin Hanson.
\newblock Combinatorial information market design.
\newblock \emph{Information Systems Frontiers}, 5\penalty0 (1):\penalty0
  107--119, 2003.

\bibitem[Hanson(2012)]{hanson2012logarithmic}
Robin Hanson.
\newblock Logarithmic markets coring rules for modular combinatorial
  information aggregation.
\newblock \emph{The Journal of Prediction Markets}, 1\penalty0 (1):\penalty0
  3--15, 2012.

\bibitem[Jackson(2010)]{jackson2010social}
Matthew~O Jackson.
\newblock \emph{Social and economic networks}.
\newblock Princeton university press, 2010.

\bibitem[Kong and Schoenebeck(2018)]{KongS2018AliceBobAlic}
Yuqing Kong and Grant Schoenebeck.
\newblock {Optimizing Bayesian Information Revelation Strategy in Prediction
  Markets: the Alice Bob Alice Case}.
\newblock In \emph{9th Innovations in Theoretical Computer Science Conference
  (ITCS 2018)}, volume~94 of \emph{Leibniz International Proceedings in
  Informatics (LIPIcs)}, pages 14:1--14:20, Dagstuhl, Germany, 2018. Schloss
  Dagstuhl--Leibniz-Zentrum fuer Informatik.

\bibitem[Kong and Schoenebeck(2019)]{Kong:2019:ITF:3309879.3296670}
Yuqing Kong and Grant Schoenebeck.
\newblock An information theoretic framework for designing information
  elicitation mechanisms that reward truth-telling.
\newblock \emph{ACM Trans. Econ. Comput.}, 7\penalty0 (1):\penalty0 2:1--2:33,
  January 2019.
\newblock ISSN 2167-8375.

\bibitem[Lobel and Sadler(2016)]{lobel2016preferences}
Ilan Lobel and Evan Sadler.
\newblock Preferences, homophily, and social learning.
\newblock \emph{Operations Research}, 64\penalty0 (3):\penalty0 564--584, 2016.

\bibitem[Mossel and Schoenebeck(2010)]{MosselS10}
Elchanan Mossel and Grant Schoenebeck.
\newblock Arriving at consensus in social networks.
\newblock In \emph{The First Symposium on Innovations in Computer Science (ICS
  2010)}, January 2010.

\bibitem[Parikh and Krasucki(1990)]{parikh1990communication}
Rohit Parikh and Paul Krasucki.
\newblock Communication, consensus, and knowledge.
\newblock \emph{Journal of Economic Theory}, 52\penalty0 (1):\penalty0
  178--189, 1990.

\bibitem[Schoenebeck and Yu(2018)]{schoenebeck2018consensus}
Grant Schoenebeck and Fang-Yi Yu.
\newblock Consensus of interacting particle systems on erd{\"o}s-r{\'e}nyi
  graphs.
\newblock In \emph{Proceedings of the Twenty-Ninth Annual ACM-SIAM Symposium on
  Discrete Algorithms}, pages 1945--1964. SIAM, 2018.

\bibitem[Shannon(1948)]{shannon1948mathematical}
Claude~Elwood Shannon.
\newblock A mathematical theory of communication.
\newblock \emph{The Bell system technical journal}, 27\penalty0 (3):\penalty0
  379--423, 1948.

\bibitem[Smith and Sorensen(2000)]{Smith2000-fi}
Lones Smith and Peter Sorensen.
\newblock Pathological outcomes of observational learning.
\newblock \emph{Econometrica}, 68\penalty0 (2):\penalty0 371--398, 2000.

\bibitem[Ting(1962)]{ting1962amount}
Hu~Kuo Ting.
\newblock On the amount of information.
\newblock \emph{Theory of Probability \& Its Applications}, 7\penalty0
  (4):\penalty0 439--447, 1962.

\bibitem[Yildiz et~al.(2013)Yildiz, Ozdaglar, Acemoglu, Saberi, and
  Scaglione]{yildiz2013binary}
Ercan Yildiz, Asuman Ozdaglar, Daron Acemoglu, Amin Saberi, and Anna Scaglione.
\newblock Binary opinion dynamics with stubborn agents.
\newblock \emph{ACM Transactions on Economics and Computation (TEAC)},
  1\penalty0 (4):\penalty0 1--30, 2013.

\end{thebibliography}
